\documentclass[11pt,letter]{article}

\usepackage{amsmath}
\usepackage{amssymb}
\usepackage{amsthm}

\usepackage{epsfig}
\usepackage{multicol}

\usepackage{algorithmic}
\usepackage{algorithm}

 \newtheorem{prop}{Proposition}[section]
 \newtheorem{theorem}[prop]{Theorem}

 \newtheorem{lemma}[prop]{Lemma}
 
 \newtheorem{proposition}[prop]{Proposition}

\newtheorem{corollary}[prop]{Corollary}

\theoremstyle{definition}
 \newtheorem{definition}[prop]{Definition}

\newcommand{\pref}{\succ}

\newcommand{\thetatwo}{\Theta_2^{\mathrm{p}}}

\def\calA{{\mathcal{A}}}
\def\calB{\mathcal{B}}
\def\calF{\mathcal{F}}
\def\calG{\mathcal{G}}

\def\calS{\mathcal{S}}

\def\calP{\mathcal{P}}
\def\calR{\mathcal{R}}
\def\calT{\mathcal{T}}
\def\poly{\mathrm{poly}}

\def\vect{\mathbf{t}}
\def\vecs{\mathbf{s}}
\def\vecr{\mathbf{r}}
\def\vecd{\mathbf{d}}

\def\vecw{\mathbf{w}}
\def\vecalpha{{\vec\alpha}}
\def\eps{\varepsilon}

\newcommand{\comment}[1]{}

\newcommand{\shift}{\ensuremath{\mathit{shf}}}
\newcommand{\gain}{\ensuremath{\mathit{gain}}}
\newcommand{\bestbuy}{\ensuremath{\mathit{buy}}}

\newcommand{\p}{{\rm P}}
\newcommand{\np}{{\rm NP}}

\newcommand{\score}{{\ensuremath{\mathit{Sc}}}}

\newcommand{\opt}{{\ensuremath{\mathrm{opt}}}}

\newcommand{\Z}{\mathbb Z}

\newcommand{\Q}{\mathbb Q}
\begin{document}
\title{Approximation Algorithms for Campaign Management}
\author{Edith Elkind\\
   Division of Mathematical Sciences\\
   Nanyang Technological University\\ Singapore
\and 
   Piotr Faliszewski \\
   Department of Computer Science\\
   AGH Univ. of Science and Technology\\ Krak\'ow, Poland
}

\maketitle

\begin{abstract}
  We study electoral campaign management scenarios 
  in which an external party can buy votes, i.e., 
  pay the voters to promote its preferred candidate in their 
  preference rankings.
  The external party's goal is to make its preferred
  candidate a winner while paying as little as possible.
  We describe a 2-approximation algorithm for this problem
  for a large class of electoral systems known as 
  scoring rules. Our result holds even for weighted 
  voters, and has applications for campaign management 
  in commercial settings.
  We also give approximation algorithms for our problem
  for two Condorcet-consistent rules, namely, the Copeland
  rule and maximin.
\end{abstract}

\section{Introduction}\label{sec:intro}
Elections and voting play an important role in the functioning
of the modern society. In the standard model of voting, each voter's
preferences are represented by a total order over the alternatives
(candidates), and some voting rule is used to determine the election 
winner(s). However, in practice, the voters' preferences are often
flexible, and it is possible to affect the outcome of the election
by campaigning for or against a certain candidate. Indeed,   
campaign management is a multi-million dollar industry, and there is 
overwhelming evidence that the amount of money invested into
a candidate's campaign is strongly correlated with her
chances of winning the election.

The notion of {\em bribery} proposed by Faliszewski,
Hemaspaandra, and Hemaspaandra~\cite{fal-hem-hem:j:bribery}
can be viewed as a formal model of electoral campaign management.
In the model of~\cite{fal-hem-hem:j:bribery}, 
each (possibly weighted) voter is associated with a certain price, and, by
paying the price, the briber can change that voter's vote in any way
she likes. The briber's goal, then, is to get a particular candidate
elected, subject to a budget constraint. To connect this description
with our original campaign management scenario, observe that bribing a 
(weighted) voter can be interpreted as mounting an election campaign targeted
at a particular group of voters with identical preferences.

However, this interpretation does not take into account 
that in practice it may be relatively easy to convince a voter
to make small changes to his vote, but hard or impossible
to convince him to adopt an entirely new preference ordering.
To remedy this, several subsequent 
papers~\cite{fal:c:nonuniform-bribery,fal-hem-hem-rot:j:llull,elk-fal-sli:c:swap-bribery}
allow the briber to modify the voters' preferences in a more fine-grained manner. 
In~\cite{fal:c:nonuniform-bribery} and~\cite{fal-hem-hem-rot:j:llull}, 
the authors depart from the assumption that the voters
are represented by the preference orders.
We will therefore focus on the framework of~\cite{elk-fal-sli:c:swap-bribery},
which operates in the standard model of voting, and
assumes that the briber can pay each voter to swap any two candidates 
that are adjacent in that voter's ordering; this type of bribery is called 
{\em swap bribery}. In the context of campaign management, such a swap
corresponds to an ad that compares two particular candidates. 
A special case of swap bribery that was also
suggested in~\cite{elk-fal-sli:c:swap-bribery} is {\em shift
bribery}, where the briber is limited to buying swaps that
involve her preferred candidate; in effect, this is equivalent to
allowing the briber to shift her preferred candidate up in the voters'
preference orderings. The constraint that a campaign ad should
involve the briber's preferred candidate is very natural
from the ethics perspective; as we will see later, it also leads
to more tractable computational problems.

The complexity-theoretic study of swap and shift bribery was initiated
in~\cite{elk-fal-sli:c:swap-bribery}, where the authors show that the
associated computational problem is hard for many voting rules.
However, campaign management can be naturally viewed as an
optimization problem, and hence we can approach it using the framework
of approximation algorithms. This line of research was first suggested
in~\cite{elk-fal-sli:c:swap-bribery}, where the authors give a
2-approximation algorithm for shift bribery under the Borda rule. 
We expand the study of approximation algorithms for
campaign management to voting rules other than Borda, 
and to weighted voters.

Our main result is a 2-approximation algorithm for shift bribery under
all scoring rules (a large class of voting rules, which includes
Borda); our result holds even for weighted voters.  Under a scoring
rule, each candidate gets a certain number of points from each voter,
which is determined by that candidate's position in the voter's
preferences, and the winner is the candidate with the maximum number
of points. Unlike most of the existing algorithms for scoring rules
(see, e.g.,~\cite{fal-hem-hem:j:bribery}), our algorithm does not
assume that the number of candidates is constant, but rather accepts
the scoring vector as an input. Our proof has an unusual structure:
we first design a pseudopolynomial 2-approximation
algorithm for our problem, then convert it into a
$(2+\eps)$-approximation scheme, and finally
turn the $(2+\eps)$-approximation scheme into a 2-approximation
algorithm using a bootstrapping argument.

Interestingly, shift bribery under scoring rules provides mathematical
framework for campaign management scenarios that are not related to elections.
Consider, for example, an advertiser in a sponsored search setting
who wants to ensure that his ads get more clicks than those of the competitors, 
and is willing to make an additional investment in his campaign
to achieve that.
By associating the competing ads with candidates, search terms
with (weighted) voters, and scores for position $i$ 
with clickthrough rates for an ad in position $i$, 
we can reduce the advertiser's problem to shift bribery with weighted voters.
This example suggests that our 2-approximation algorithm can be used
for campaign management in a variety of settings, including---but not limited 
to---voting.

We also give approximation algorithms for shift bribery under two
voting rules that have the attractive property of {\em Condorcet
  consistency}, namely, the Copeland rule and maximin.\footnote{We
  mention that an earlier version of this paper also included results
  on Bucklin voting. Unfortunately, the proofs contained errors and
  thus we removed the section on Bucklin from the paper.}  To
complement our positive results, we show that, in contrast to shift
bribery, swap bribery is usually hard to approximate. We conclude the
paper by suggesting directions for future work.

\section{Preliminaries}

In this section we describe relevant notions from computational social
choice and define the shift bribery problem.  We take $\Z^+$ to be the
set of all nonnegative integers.

\medskip
\noindent
\textbf{Elections.}\quad An {\em election} is a pair $E = (C,V)$,
where $C = \{c_1, \ldots c_m\}$ is the set of {\em candidates} and $V
= (v^1, \ldots, v^n)$ is a collection of {\em voters}. Each voter
$v^i$ is described by her {\em preference order} $\pref^i$, which is a
strict linear order over $C$: $c\pref^i c'$ means that voter $v^i$
prefers $c$ to $c'$. 
We will also consider settings where each voter $v^i$ has a {\em weight}
$w_i$; in this case, her vote is interpreted as $w_i$ votes.

A {\em voting rule} is
a function that given an election $E = (C,V)$ outputs a set 
$W \subseteq C$ of {\em election winners}. Note that we do not
require the voting rule to produce a unique winner, i.e., we
work in the so-called {\em nonunique-winner model}. This approach 
is standard in computational social choice literature, as it allows
to define classic voting rules (see below) in a more natural manner.
In practice, this means that a voting rule may have to be combined
with a tie-breaking rule.

\smallskip
\noindent
\textbf{Voting Rules.}\quad
We will now describe several well-known voting rules that 
will be considered in this paper.
All voting rules listed below are defined 
for an election $E = (C,V)$ with $C = \{c_1, \ldots, c_m\}$, 
$V = (v^1, \ldots, v^n)$. 
For all rules defined in terms of scores (points),   
the winner(s) are the candidate(s) with the maximum score 
(highest number of points).

\begin{description}
\item[Scoring rules.] A {\em scoring rule} $\calR_\vecalpha$
  is described by a vector $\vecalpha = (\alpha_1, \ldots, \alpha_m)$, 
  where $\alpha_i\in\Z^+$ for $i=1, \dots, n$, 
  and $\alpha_1\ge\dots\ge \alpha_m$.
  Under $\calR_\vecalpha$, each
  candidate $c_i$ receives $\alpha_j$ points from each voter that
  ranks him in the $j$-th position. 
  Note that each scoring rule is defined for a
  fixed number of candidates. Thus, we often consider voting rules
  that are defined by {\em families} of scoring rules $(\vecalpha^m)_{m=1, 2, \dots}$, 
  with one vector for each number of candidates. In particular, 
  {\em Borda}
  is the rule given by $\alpha^m_j=m-j$ for $j=1, \dots, m$, 
  and {\em $k$-approval} is the rule given by 
  $\alpha^m_j=1$ for $j\le k$, $\alpha^m_j=0$ for $j>k$;
  $1$-approval is also known as {\em plurality}.

\item[Condorcet consistent rules.] For any $c_i, c_j \in
  C$, let $N_E(c_i,c_j)$ denote the number of voters in $E$ who
  prefer $c_i$ to $c_j$. If $N_E(c_i,c_j) > N_E(c_j,c_i)$, then we say
  that $c_i$ wins the {\em pairwise election} against $c_j$.
  A candidate $c\in C$ is called the {\em Condorcet winner}
  if he wins the pairwise elections against all other candidates in $C$.
  Note that some elections may not have a Condorcet winner.
  We say that a voting rule $\calR$ is {\em Condorcet-consistent} 
  if for any election $E$ that has a Condorcet winner $c$ we have $\calR(E)=\{c\}$.
  Two examples of Condorcet-consistent rules are
  Copeland and maximin, defined as follows. 
  For any rational $\alpha\in[0, 1]$, 
  {\em Copeland\,$^\alpha$} grants one point to a candidate $c_i\in C$ 
  for each pairwise elections that $c_i$ wins,
  and $\alpha$ points for each pairwise election that $c_i$ ties. 
  The {\em maximin score} of $c_i$ is the number of votes that $c_i$ receives
  in her worst pairwise election, i.e., 
  $\min_{c_j \in C\setminus\{c_i\}}N_E(c_i,c_j)$.

\end{description}
We denote by $\score^\calR_E(c)$ the
$\calR$-score of a candidate $c \in C$ in an election $E=(C, V)$;
we omit the superscript $\calR$ when the voting rule 
is clear from the context.

\medskip
\noindent
\textbf{The Shift Bribery Problem.}\quad
This section is based on the definitions 
from~\cite{elk-fal-sli:c:swap-bribery}.
Consider an election $E =
(C,V)$ with $C = \{p, c_1, \ldots, c_{m-1}\}$, $V = (v^1, \ldots,
v^n)$. Suppose that our goal is to ensure that the designated candidate $p$ is
a winner of the election under a voting rule $\calR$. In order to achieve
this goal, we can ask
each voter $v^i$ to shift $p$ upwards in her vote by a certain number
of positions. This models the fact that we can campaign in favor of
$p$. However, each such shift has a cost. Specifically, each
voter $v^i$ has a {\em cost function} $\pi^i:\Z^+\rightarrow\Z^+$, 
where $\pi^i(k)$, $k \in \Z^+$, is the cost of shifting $p$ upwards 
by $k$ positions in $\pref^i$.
We require that each $\pi^i$, $i=1, \dots, n$, satisfies
$\pi^i(0) = 0$ and $\pi^i(k)\le \pi^i(k+1)$ for $k\in\Z^+$.
Also, when $v^i$ ranks $p$ in position $t$, 
we require $\pi^i(s) = \pi^i(t-1)$ for all $s\ge t$; thus, 
the function $\pi^i$ is fully specified by its values at $1, \dots, t-1$.
Note that we assume that $\pi^i(k)<\infty$ for all    
$i=1, \dots, n$ and $k\in\Z^+$. However, all
our proofs can be generalized to the case where
$\pi^i$ can be $+\infty$ (i.e., some voters
cannot be bribed to move $p$ by more than some given number of positions).
We seek an action that makes $p$ a winner at the
minimum cost. %

\begin{definition}
  Let $\calR$ be a voting rule. An instance of {\sc $\calR$-shift-bribery} problem
  is a tuple $I = (C, V,\Pi, p)$, where
  $C = \{p,c_1, \ldots, c_{m-1}\}$, 
  $V = (v^1, \ldots, v^n)$ is a collection of preference orders over $C$, 
  $\Pi = (\pi^1, \ldots, \pi^n)$ is a family of cost functions, 
  and $p\in C$ is a designated candidate. 
  The goal is to find a minimal value $b$ such that there is a sequence 
  $\vect = (t_1, \ldots, t_n) \in (\Z^+)^n$ with the following properties:
  (a) $b = \sum_{i=1}^{n}\pi^i(t_i)$, and (b) if for each $i=1, \dots, n$
  we shift $p$ upwards in the $i$-th vote by $t_i$
  positions, then $p$ becomes an $\calR$-winner of $E$.
  We denote this value of $b$ by $\opt(I)$.
\end{definition}

In {\sc weighted $\calR$-shift-bribery}, the description of the instance includes
a vector of voters' weights $\vecw=(w_1, \dots, w_n)$, i.e., we have
$I = (C, V,\Pi, p, \vecw)$. 

We will call the sequence $\vect = (t_1, \ldots, t_n)$ a
{\em shift-action}.  Let $\shift(C,V,\vect)$ denote the election obtained
from $(C,V)$ by shifting $p$ upwards by $t_i$
positions in the $i$-th vote (or placing $p$ on top of that vote, if 
$v^i$ ranks $p$ in position $t<t_i+1$ before the bribery). 
A shift-action is {\em successful} if $p$ is a winner of
$\shift(C,V,\vect)$.  Additionally, let $\Pi(\vect)=\sum_{i=1}^n\pi^i(t_i)$.

Let $I = (C,V,\Pi,p)$ be an instance of $\calR$-\textsc{shift-bribery}
and let $\vect = (t_1, \ldots, t_n)$. Overloading notation, we let
$\shift(I,\vect)$ denote an instance 
$\hat{I} = (C,\hat{V},\hat{\Pi},p)$ of $\calR$-\textsc{shift-bribery}
given by
(a) $(C,\hat{V}) = \shift(C,V,\vect)$, and
(b) $\hat{\Pi} = (\hat{\pi}^1, \ldots, \hat{\pi}^n)$ where for each $i=1, \dots, n$
   we have $\hat{\pi}^i(k) = \pi^i(k+t_i)-\pi^i(t_i)$.
That is, $\shift(I,\vect)$ represents the instance of shift bribery 
obtained from $I$ by applying the shift-action $\vect$; the costs are modified
to reflect the fact that some shifts have already been performed.

Given an instance $I$ of $\calR$-\textsc{shift-bribery}
or {\sc weighted} $\calR$-\textsc{shift-bribery},   
we denote by $|I|$ the representation size of $I$ assuming that 
all entries of $\Pi$ (and $\vecw$, for the weighted case) 
are given in binary.
Similarly, $|\vecalpha|$ denotes the number of bits
in the binary encoding of a scoring vector $\vecalpha$.

\section{Scoring Rules}\label{sec:scoring}
In this section, we describe a 2-approximation algorithm that works
for all scoring rules. 

\begin{theorem}\label{thm:scoring2}
  There is an algorithm $\calB$ that given a
  scoring rule $\vecalpha = (\alpha_1, \ldots, \alpha_m)$ and an
  instance $I = (C,V,\Pi,p)$ of $\calR_\vecalpha$-\textsc{shift-bribery} 
  with $|C| = m$, outputs 
  a successful shift-action $\vect$ for $I$ 
  that satisfies $\Pi(\vect) \leq 2\opt(I)$, and runs
  in time $\poly(|I|,|\vecalpha|)$.
\end{theorem}

We split the proof of Theorem~\ref{thm:scoring2} into three steps.
First (Proposition~\ref{prop:pseudo}) we describe a pseudopolynomial 
$2$-approximation algorithm $\calA$ for our problem. Then 
(Proposition~\ref{prop:approx}) we use 
$\calA$ to construct another algorithm $\calA'$, which for any $\eps>0$
produces a $(2+\eps)$-approximation and runs in time polynomial
in the instance size and $\frac{1}{\eps}$. Finally, we convert
$\calA'$ into a $2$-approximation algorithm by bootstrapping. 
Throughout the proof, we 
fix a scoring rule $\calR_\vecalpha = (\alpha_1, \ldots, \alpha_m)$.

We first prove a preliminary lemma that will be used to
demonstrate the correctness of our approach.

  \begin{lemma}\label{lem:any2}
    Let $\vecs = (s_1, \ldots, s_n)$ be a successful shift-action for 
    $(C, V, \Pi, p)$, 
    and let $k= \score_{\shift(C, V,\vecs)}(p) - \score_{(C, V)}(p)$.
    Then every shift-action $\vecr = (r_1, \ldots, r_n)$ such that  
    $\score_{\shift(C, V,\vecr)}(p) = \score_{(C, V)}(p)+2k$ 
    is successful for $(C, V, \Pi, p)$.
  \end{lemma}
  \begin{proof}
    When $p$ is shifted
    from position $i+1$ to position $i$ in some vote $\pref^j$, he obtains
    $\alpha_i-\alpha_{i+1}$ extra points, while the candidate $c$ that
    was in position $i$ in $\pref^j$ prior to the shift loses
    $\alpha_i-\alpha_{i+1}$ points; the scores of all other candidates
    remain unchanged. Since $\vecs$ increases $p$'s score by $k$
    and $p$ wins in $\shift(C, V, \vecs)$, we have 
    $\max_{c \in C}[\score_{(C, V)}(c)-\score_{(C, V)}(p)]\leq 2k$. 
    Now, $\vecr$
    increases $p$'s score by $2k$ points, and does not increase the score
    of any other candidate, so the lemma follows.
  \end{proof}

We are now ready to implement the first step of our plan. The
algorithm $\calA$ presented in the next proposition is inspired by 
the 2-approximation algorithm
for the Borda rule that appears in~\cite{elk-fal-sli:c:swap-bribery};
however, its analysis is substantially different.

\begin{proposition}\label{prop:pseudo} 
  There exist an algorithm $\calA$ that given                     
  a scoring rule $\vecalpha = (\alpha_1, \ldots, \alpha_m)$ and an
  instance $I = (C,V,\Pi,p)$ of $\calR_\vecalpha$-\textsc{shift-bribery}
  with $|C| = m$, outputs
  a successful shift-action $\vect$ for $I$
  that runs in time $\poly(|I|,|\vecalpha|, \sum_{i=1}^n\pi^i(m))$
  and satisfies $\Pi(\vect) \leq 2\opt(I)$.
\end{proposition}
\begin{proof}
  Consider an instance $I = (C,V,\Pi,p)$ of
  $\calR_\vecalpha$-\textsc{shift-bribery} such that $C = \{p,c_1, \ldots,
  c_{m-1}\}$ and $V = (v^1, \ldots, v^n)$, and set $E=(C, V)$.
  For each integer $\ell \geq 0$, and each instance $J =
  (C,V',\Pi',p)$ of $\calR_\vecalpha$-\textsc{shift-bribery}, 
  set $\bestbuy(J,\ell)=\shift(C, V', \vect^\ell)$, where
  $\vect^\ell$  is chosen so as to maximize $p$'s score
  subject to the constraint $\Pi'(\vect^\ell)\le \ell$, i.e.,  
  $\vect^\ell\in\arg\max\{
    \score_{\shift(C, V', \vect)}(p)\mid \Pi'(\vect) \leq\ell
    \}$.

  \begin{lemma}\label{lem:bestbuy}
    For any $\ell\ge 0$, the instance $\bestbuy(I,\ell)$ is
    computable in time $\poly(|I|, \ell)$.
  \end{lemma}
  \begin{proof}
    For $i=1, \dots, n$, $k=0, \dots, m$, 
    let $\gain^i(k)$ denote
    the number of points that $p$ gains from being shifted upwards
    by $k$ positions in $\pref^i$ (or, to the top
    position in $\pref^i$, if $k$ is too large).

    Our algorithm relies on dynamic programming.  For each 
    $i=0, \dots, n$, $j=0, \dots, \ell$, let $F(i,j)$ denote the
    maximum increase in $p$'s score that can be achieved by spending exactly
    $j$ dollars on bribing the first $i$ voters.
    If it is impossible to spend exactly $j$ dollars 
    on bribing the first $i$ voters, set $F(i,j) = -\infty$.
    Naturally, we have $F(0,0) = 0$ and $F(0,j) = -\infty$ for 
    $j=1, \dots, \ell$.  Further, it is easy to see that for
    each $i=1, \dots, n$ and $j=0, \dots, \ell$ we have
    $$
    F(i,j) = \max\{F(i-1,j-\pi^i(k))+\gain^i(k)\mid 
              k =0, \ldots, m, j \geq \pi^i(k)\}.
    $$
    Thus, we can compute $F(n,j)$  
    for all $j=0, \dots, \ell$ in time $\poly(|I|, \ell)$.
    Now pick a value $j\in\{0, \dots, \ell\}$ that maximizes
    $F(n, j)$.
    Using standard techniques, we can 
    find a shift-action $\vect$ that corresponds to $F(n,j)$.
    Clearly, we have $\bestbuy(I,\ell) = \shift(C, V, \vect)$.
    \renewcommand{\qedsymbol}{(Lemma~\ref{lem:bestbuy})~$\square$}
\end{proof}

  \noindent
  The pseudocode for algorithm $\calA$ is given in
  Figure~\ref{fig:algA}. Lemma~\ref{lem:bestbuy} implies
  that the running time of $\calA$ is polynomial in
  $|I|$ and $\sum_{i=1}^{n}\pi^i(m)$. It remains to show that 
  $\calA$ indeed produces a $2$-approximate solution.

  \begin{figure}
    \begin{center}
      \begin{tabbing}
        123\=123\=123\=123\=123\=\kill
        \> \textbf{procedure} $\calA(C, V, \Pi, p)$\\
        \> \textbf{begin}\\
        \>\> Set $m = |C|$, $n = |V|$, $M = \sum_{i=1}^n\pi^i(m)$, $b = \infty$; \\
        \>\> \textbf{for} $\ell_1 = 0$ \textbf{to} $M$ \textbf{do}\\
        \>\>\> \textbf{for} $\ell_2 = 0$ \textbf{to} $M$ \textbf{do}\\
        \>\>\> \textbf{begin}\\
        \>\>\>\>  $I' = \bestbuy(I,\ell_1)$; \\
        \>\>\>\>  $I'' = \bestbuy(I',\ell_2)$; \\
        \>\>\>\>  \textbf{if} $p$ is an $\calR_{\vecalpha}$-winner in $I''$ 
                               \textbf{and} $\ell_1+\ell_2 < b$ \textbf{then}\\
        \>\>\>\>\>  set $b = \ell_1+\ell_2$;\\
        \>\>\> \textbf{end}\\
        \>\> \textbf{return} $b$; \\
        \> \textbf{end}
      \end{tabbing}
      \caption{\label{fig:algA}Algorithm $\calA$.}
    \end{center}
  \end{figure}

    To this end, 
    we will show that there exist
    $\ell_1,\ell_2 \leq \sum_{i=1}^n\pi^i(m)$ such that 
    $I' = \bestbuy(I,\ell_1)$, $I'' = \bestbuy(I',\ell_2)$, 
    $p$ is an $\calR_{\vecalpha}$-winner in $I''$, 
    and $\ell_1+\ell_2\le 2\opt(I)$.

    Let $\vect = (t_1, \ldots, t_n)$ be an optimal shift-action that
    ensures $p$'s victory, that is, $\Pi(\vect) = \opt(I)$. 
    Set $k = \score_{\shift(C, V, \vect)}(p) - \score_E(p)$.

    Consider an instance $I'$ obtained from $I$ by spending the total cost
    of the optimal shift action greedily, 
    i.e., so as to maximize $p$'s score.
    Formally, let $\ell_1 = \Pi(\vect)$ and set $I' = \bestbuy(I,\ell_1)$. 
    Let $\vecs' = (s'_1, \ldots, s'_n)$ be the shift-action that transforms $I$
    into $I' = (C,V',\Pi',p)$, and set $E' = (C,V')$. 
    By construction, we have $\score_{E'}(p) \ge \score_E(p)+k$.

    Let $\vecr=(r_1, \ldots, r_n)$ be the common part of shift-actions 
    $\vect$ and $\vecs'$, i.e.,
    set $r_i = \min\{t_i,s'_i\}$ for $i=1, \dots, n$.  
    Let $I^r = \shift(C, V,\vecr)$,
    where $I^r = (C,V^r,\Pi^r,p)$, and set $E^r = (C,V^r)$. 
 
    Finally, set $\ell_2 = \Pi'(\vect-\vecr)$, $I'' =
    \bestbuy(I',\ell_2)$, and let $\vecs'' = (s''_1, \ldots, s''_n)$
    be the shift-action that transforms $I'$ into $I'' = (C,V'',
    \Pi'',p)$. Let $E'' = (C,V'')$.  Observe that for each $i=1,
    \dots, n$ we have either $t_i-r_i=0$, in which case $\pi'^i(t_i-r_i)=0$, 
    or $t_i-r_i=t_i-s'_i$, in which case $\pi'^i(t_i-r_i) =
    \pi'^i(t_i-s'_i)=\pi^i(t_i)-\pi^i(s'_i)$.  Therefore, we have
    $\Pi'(\vect-\vecr)\le \Pi(\vect)$.  Now, the total cost of
    $\vecs'+\vecs''$ is given by $\ell_1+\ell_2 =
    \Pi(\vect)+\Pi'(\vect-\vecr)\le 2\Pi(\vect)$.  As
    $\Pi(\vect)=\opt(I)$, we obtain $\ell_1+\ell_2\le 2\opt(I)$.  It
    remains to show that $p$ is a winner in $\shift(C, V,
    \vecs'+\vecs'')$.

    Set $k^r= \score_{E^r}(p) - \score_E(p)$.
    The shift-actions $\vect-\vecr$ and $\vecs'-\vecr$ satisfy
    \begin{eqnarray}
      \label{eq:scring2-1}
      \score_{\shift(C, V^r,\vect-\vecr)}(p)    &=& \score_{E^r}(p)+(k-k^r), \\
      \label{eq:scring2-2}
      \score_{\shift(C, V^r,\vecs'-\vecr)}(p) &\geq& \score_{E^r}(p)+(k-k^r).
    \end{eqnarray}
    We have $\shift(C, V^r,\vect-\vecr)=\shift(C, V, \vect)$, so
    $p$ is an $\calR_\vecalpha$-winner in $\shift(C, V^r,\vect-\vecr)$.
    Thus, by Lemma~\ref{lem:any2}, 
    any shift-action that increases the score of $p$ in $E^r$ by
    $2(k-k^r)$ points ensures that $p$ is a winner in the resulting
    election. We will now show that this holds 
    for the shift-action $\vecs''+(\vecs'-\vecr)$, 
    and hence $p$ is a winner in
    $\shift(C, V^r,\vecs''+\vecs'-\vecr)=\shift(C, V, \vecs''+\vecs')$.

    For each $i=1, \dots, n$, if $t_i-r_i \neq 0$, then $r_i=s'_i$
    and the $i$'th voter ranks $p$ in the same position both in $V'$ and in $V^r$.
    Thus, $\Pi^r(\vect-\vecr) = \Pi'(\vect-\vecr) = \ell_2$, and applying 
    $\vect-\vecr$ to $I'$ increases $p$'s score by the same amount as 
    applying $\vect-\vecr$ to $I^r$.
    By equation~\eqref{eq:scring2-1}, this implies
    \begin{equation}
      \label{eq:scring2-3}
      \score_{\shift(C, V',\vect-\vecr)}(p) = \score_{E'}(p) + (k-k^r).
    \end{equation}

    By definition, $\vecs''$ is a shift-action of cost at most $\ell_2 =
    \Pi'(\vect-\vecr)$ that applied to $E'$ increases $p$'s score
    as much as possible. Thus, equation~\eqref{eq:scring2-3} implies
    \begin{equation}
      \label{eq:scring2-4}
      \score_{E''}(p) \geq \score_{E'}(p) + (k-k^r).
    \end{equation}
    Since $E'' = \shift(E',\vecs'')$ and $E' = \shift(E^r,\vecs'-\vecr)$, by
    combining equations~\eqref{eq:scring2-2} and~\eqref{eq:scring2-4}
    we obtain the following inequality:
    \begin{eqnarray*}
      \score_{E''}(p) &\geq& \score_{E'}(p) + (k-k^r) \\
      & = & \score_{\shift(C, V^r,\vecs'-\vecr)}(p) + (k-k^r) \\
      & \geq & \score_{E^r}(p)+2(k-k^r).
    \end{eqnarray*}
    Thus, $p$ is a winner in election $E''$. 
This completes the proof of Proposition~\ref{prop:pseudo}.
\end{proof}

We will now convert algorithm $\calA$
into a $(2+\eps)$-approximation scheme.
A natural approach is to scale the bribery price functions 
given by $\Pi$ by a sufficiently large parameter $K$. However, the 
appropriate choice of $K$ depends on the actual value of the optimal
solution. Therefore, we make polynomially many guesses
of the value of $K$.

\begin{proposition}\label{prop:approx}
  There exist an algorithm $\calA'$ that given a rational $\eps>0$,
  a scoring rule $\vecalpha = (\alpha_1, \ldots, \alpha_m)$ and an
  instance $I = (C,V,\Pi,p)$ of $\calR_\vecalpha$-\textsc{shift-bribery}
  with $|C| = m$,  
  runs in time  $\poly(|I|,|\vecalpha|,\frac{1}{\eps})$
  and outputs a successful shift-action $\vect$ for $I$
  that satisfies $\Pi(\vect) \leq (2+\eps)\opt(I)$.
\end{proposition}
\begin{proof}
  Let $R = \max_{i\in \{1, \ldots, n\}}\pi^i(m)$.
  Algorithm $\calA'$ executes $\lceil \log R \rceil$ iterations, 
  each with a different value $\rho$, the current
  estimate of the cost of the most expensive shift. We start with
  $\rho = 1$ and we double it after every iteration. In each iteration
  we compute some solution---using algorithm $\calA$---and then either
  discard it or keep it. After the last iteration, we pick a solution
  with the lowest cost. Let us now describe a single iteration.

  Let $\rho$ be the current guess of the cost of the most expensive
  shift. Let $K = \frac{\rho\eps}{n}$.  We construct an
  instance $J = (C,V,\Lambda,p)$, where $\Lambda = (\lambda_1,
  \ldots, \lambda_n)$ is a sequence of shift-bribery cost functions
  defined as follows. For each $i=1, \dots, n$, 
  $j=0, \dots, m$, we set
  \[
  \lambda^i(j) = \left\{ \begin{array}{ll}
      \left\lceil \frac{\pi^i(j)}{K} \right\rceil & \mbox{ if $\pi^i(j) \leq \rho$,} \\
      2(\frac{n^2}{\varepsilon}+n)+1 & \mbox{ otherwise.}
    \end{array}\right.
  \]
  Note that all the values of functions in $(\lambda_1, \ldots,
  \lambda_n)$ are polynomially bounded in $n$ and
  $\frac{1}{\varepsilon}$. Algorithm $\calA$ solves $J$ in time
  $\poly(n,\frac{1}{\varepsilon})$, producing a
  shift-action $\vecs = (s_1, \ldots, s_n)$.  Note
  that since $I$ and $J$ differ only in cost functions, $\vecs$ is a
  valid solution for $I$.  If we have
  $\lambda^i(t_i) \leq 2(\frac{n^2}{\varepsilon}+n)$ (equivalently,
  $\pi^i(j) \leq \rho$) for all $i=1, \dots, n$, we store solution 
  $\vecs$ for further use. Otherwise we discard it. 
  Observe that in the last iteration we have $\rho \geq R$, so the solution
  obtained in that iteration is not discarded.
  After all iterations have
  been computed, we output the cheapest stored solution. 
  Clearly, algorithm $\calA'$ runs in polynomial time. It
  remains to show that it produces an accurate approximation.

  Let $\vect = (t_1, \ldots, t_n)$ be a shift-action that corresponds
  to an optimal solution, and let $\sigma$ be the cost of the most expensive shift within $\vect$, 
  i.e., set $\sigma = \max_{i=1, \ldots, n}\pi^i(t_i)$.
  Consider an
  iteration in which $\frac{\rho}{2} \leq \sigma \leq \rho$.  Let $J =
  (C,V,\Lambda,p)$ be the scaled-down instance used in this iteration.
  Note that $\Lambda(\vect) \leq n\left\lceil\frac{\sigma}{K}\right\rceil
  \leq n\left\lceil\frac{\rho}{K}\right\rceil = n \left\lceil
    \frac{n}{\varepsilon} \right\rceil \leq
  \frac{n^2}{\varepsilon}+n$.  Thus, algorithm $\calA$ finds a
  solution for $J$ with cost (in terms of $\Lambda$) at most
  $2(\frac{n^2}{\varepsilon}+n)$, and so this solution certainly is
  not discarded.

  Let $\vecs = (s_1, \ldots, s_n)$ be the solution produced for $J$ by
  $\calA$. As argued before, $\vecs$ is also a solution for $I$.  We have
   $\Pi(\vect) \leq \Pi(\vecs) \leq K\Lambda(\vecs) \leq 2K\Lambda(\vect)$.
  Indeed, the first inequality holds because $\vect$ is an optimal solution for $I$, the
  second holds due to definition of $\Lambda$, and the third one holds
  because $\vecs$ is a $2$-approximate solution for $J$.  Due to
  rounding, for each $i=1, \dots, n$ we have
  \[
  \pi^i(t_i) \leq K\lambda^i(t_i) = K\left\lceil
    \frac{\pi^i(t_i)}{K}\right\rceil \leq \pi^i(t_i)+K.
  \] 
  Thus, we obtain $K\Lambda(t) \leq \Pi(\vect)+ Kn$.  
  Using the fact that $nK =
  \rho\varepsilon$, we get $ \Pi(\vecs) \leq 2K\Lambda(\vect) \leq 2(\Pi(\vect)+
  Kn) \leq 2\Pi(\vect) + 2\rho\epsilon \leq (2+4\varepsilon)\Pi(\vect)$. 
  The last inequality follows from the fact that $\frac{\rho}{2} \leq
  \Pi(\vect)$. This completes the proof of Propositon~\ref{prop:approx}.
\end{proof}

To complete the proof of Theorem~\ref{thm:scoring2}, 
we will now convert our $(2+\varepsilon)$-approximation scheme into 
a 2-approximation algorithm using a bootstrapping argument.

\begin{proof}[Proof of Theorem~\ref{thm:scoring2}]
Let $I = (C,V,\Pi,p)$ be an instance
of {\sc $\calR_\vecalpha$-shift-bribery}, and
let $\vect = (t_1, \ldots, t_n)$ be an optimal shift-action for $I$.
By the pigeonhole principle, 
for some $i\in\{1, \dots, n\}$ we have
$
  \pi^i(t_i) \geq \frac{1}{n}\Pi(\vect).
$
Assume for now that we know $i$ and $t_i$ (subsequently, we will
show how to get rid of this assumption).

Let $\vecd = (0^{i-1},t_i,0^{m-i})$, and set $I' = \shift(I,\vecd)$.
Clearly, we have $\opt(I') = \opt(I) - \pi^i(t_i)$. 
Let $\eps=\frac{1}{n}$, 
and let $\vecs = (s_1, \ldots, s_n)$ be the shift action produced
by the algorithm $\calA'$ on $(I', \eps)$.
Clearly, $p$ is a winner in $\shift(I, \vecs+\vecd)$.
Further, by Proposition~\ref{prop:approx}, 
we have $\Pi(\vecs) \leq (2+\eps)(\Pi(\vect)-\pi^i(t_i))$.
Therefore, the cost of the shift-action
$\vecs+\vecd$ can be estimated as follows:
\begin{eqnarray*}
  \Pi(\vecs+\vecd) &=& \Pi(\vecs) + \pi^i(t_i)  \\
            &\leq& (2+\eps)(\Pi(\vect)-\pi^i(t_i)) + \pi^i(t_i) \\
            &\leq& 2\Pi(\vect) -\pi^i(t_i) + \eps\Pi(\vect)\\
            &\leq& 2\Pi(\vect) +(\eps\Pi(\vect) - \frac{1}{n}\Pi(\vect)) = 2\Pi(\vect), 
\end{eqnarray*}
where we use the fact that $\pi^i(t_i) \geq \frac{1}{n}\Pi(\vect)$.
Thus, $\vecs+\vecd$ is a $2$-approximatite solution.%

While we do not know the values of $i$ and $t_i$. there
are only $n$ possibilities for the former and 
$m$ possibilities for the latter.
Thus, our algorithm $\calB$ will try all of them,
and return the best solution.
Since $\frac{1}{\varepsilon} = n$,
the running time of $\calB$ is polynomial in $|I|$.
\end{proof}

\noindent
By using the algorithm $\calB$ with $\vecalpha=(m-1, \dots, 1, 0)$, 
we obtain a 2-approximation algorithm for the Borda rule.
This algorithm is different from the one given 
in~\cite{elk-fal-sli:c:swap-bribery}, even though they have 
the same approximation guarantee. Indeed, the algorithm
of~\cite{elk-fal-sli:c:swap-bribery} relies on a different
dynamic programming subroutine, whose running time is polynomial
in the instance size and $\sum_{i=1}^m\alpha_i$ 
(rather than the instance size and $\sum_{i=1}^n\pi^i(m)$, 
as in our construction). Since for Borda the expression
$\sum_{i=1}^m\alpha_i$ is polynomial in the size of the 
instance, this immediately produces a polynomial-time algorithm.
It is not hard
to see that the algorithm proposed in~\cite{elk-fal-sli:c:swap-bribery}
can be adapted to work for any scoring rule with $\sum_{i=1}^m\alpha_i=\poly(m)$. 
Of course, such an algorithm would be considerably faster
than the three-step procedure of Theorem~\ref{thm:scoring2}.
Indeed, one may wonder if the more complicated
algorithm described above is useful at all, since the scoring
vectors used in practice often have small coordinates.
However, an important feature of our algorithm is that 
it works even if each voter $v^i$ uses his own scoring vector
$\vecalpha^i$. This means that we can adapt it for {\em weighted} 
voters, by replacing a voter of weight $w$ with a scoring vector
$(\alpha_1, \dots, \alpha_m)$ by a unit-weight voter
with a scoring vector $(w\alpha_1, \dots, w\alpha_m)$. 

\begin{corollary}\label{cor:weights}
There is an algorithm $\calB^w$ that
given a scoring rule $\vecalpha$ and an
instance $I = (C,V,\Pi,p, \vecw)$ 
of \textsc{weighted} $\calR_\vecalpha$-\textsc{shift-bribery},  
outputs a successful shift-action $\vect$ for $I$
that satisfies $\Pi(\vect) \leq 2\opt(I)$, and runs
in time $\poly(|I|,|\vecalpha|)$.
\end{corollary}
Note that there does not seem to be an easy way to derive 
Corollary~\ref{cor:weights} from the result  
of~\cite{elk-fal-sli:c:swap-bribery}. Indeed, Corollary~\ref{cor:weights}
is quite surprising, as many problems in computational
social choice are known to be hard for weighted voters---in fact, 
\textsc{weighted shift-bribery} can be shown 
to be $\np$-complete for any nontrivial family of scoring rules.
On the other hand, an efficient algorithm for the weighted
case is very useful, as large voter weights are ubiquitous
in campaign management scenarios, where a ``voter''
corresponds to a collection of individuals that can be ``bribed''
by the same promotional activity, or in our sponsored search 
example, where the search terms may differ in popularity. 

One may also wonder if the algorithm $\calA$ can be simplified
by using a single {\bf for}-loop, which for each value of $\ell$
finds the best shift-action of cost $\ell$.
In Appendix~\ref{app:single-loop}, we explore this question
in more detail, showing that the resulting algorithm $\calG$
is never better than $\calA$, and can sometimes
produce a shift-action that is almost twice as expensive
as the one produced by $\calA$. However, we do not
know if $\calG$ is nevertheless a 2-approximation algorithm
for our problem. 

Finally, we remark that unless $\p=\np$, there is no FPTAS for shift
bribery under scoring rules (a direct consequence of the
Borda-\textsc{shift-bribery}
$\np$-hardness proof from~\cite{elk-fal-sli:c:swap-bribery}).
However, we cannot rule
out the possibility that there exist approximation
algorithms for Borda-\textsc{shift-bribery} whose
approximation ratio is less than two.

\section{Copeland and Maximin}
Paper~\cite{elk-fal-sli:c:swap-bribery} shows that the decision
version of the {\sc shift-bribery} problem is $\np$-hard for
Copeland$^\alpha$ for any rational $\alpha\in[0, 1]$
as well as for maximin.
We will now give $m$-approximation algorithms
for {\sc shift-bribery} under Copeland$^\alpha$ with rational 
$\alpha \in [0,1]$ 
and maximin. We then show how to improve
the approximation ratio for maximin to $O(\log m)$.
Finally, we argue that the $O(\log m)$ approximation ratio
for maximin is asymptotically tight, 
and show how our results can be adapted to the weighted setting.

Under Copeland and maximin, the winner is selected
on the basis of the outcomes of pairwise comparisons between the candidates.
Thus, these rules can be defined in the so-called
{\em irrational voter model}. In this model, the preferences
of each voter $i$ are given by her {\em preference table}, which is
an antisymmetric $m\times m$ matrix whose entry $(j, k)$ is $1$ 
if $i$ prefers $c_j$ to $c_k$ and $-1$ otherwise.
Clearly, a preference order can be converted into a preference table, 
but the converse is not true, as preference tables
can encode cyclic preferences.

Faliszewski et al.~\cite{fal-hem-hem-rot:j:llull} adapt the notion of
bribery to the irrational voter model; the resulting notion is known
as {\em microbribery} (see next paragraph for a formal
definition). Moreover, they show that Copeland$^\alpha$-{\sc
  microbribery} is in $\p$ for $\alpha\in\{0, 1\}$. We will first give
a poly-time algorithm for a special case of Copeland$^\alpha$-{\sc
  microbribery} for all $\alpha \in [0,1]\cap\Q$, and---based on this
result---derive an $m$-approximation algorithm for
Copeland$^\alpha$-\textsc{shift-bribery}. Then, we will argue that our
technique can be used to convert exact algorithms for {\sc
  microbribery} into $m$-approximation algorithms for {\sc shift
  bribery} for many other voting rules, including maximin.  Finally,
we will show an $O(\log m)$-approximation algorithm for {\sc
  shift-bribery} in maximin and argue it is asymptotically optimal.

First of all, we need to describe microbribery~\cite{fal-hem-hem-rot:j:llull} precisely.
Let $\calR$ be a voting rule defined in the irrational voter model, 
i.e., a mapping that for any collection of preference tables over a given
set of candidates outputs a subset of the candidates.  
In the
$\calR$-\textsc{microbribery} problem we are given an election $E = (C,V)$,
where $C = \{p,c_1, \ldots, c_{m-1}\}$ and $V = (v^1, \ldots, v^n)$ is a
collection of voters specified by their preference tables.  Also, for each 
$i=1, \dots, n$, we have a price function $\delta^i  
\colon C^2 \rightarrow \Z^+\cup\{+\infty\}$, where  
$\delta^i(c_j, c_k)$ is the cost
of flipping the $(j, k)$-th entry in the preference table of $v^i$
(i.e., if prior to the bribery we have $c_j \succ^i c_k$, 
 then after the flip we obtain $c_k \succ^i  c_j$, and vice versa).
We require the price functions to be symmetric, 
i.e., for each $i=1, \dots, n$ and all $c_j, c_k\in C$
we require $\delta^i(c_j, c_k) = \delta^i(c_k,c_j)$.
Further, we require $\delta^i(c_j, c_j)=0$ for all $i=1, \dots, n$ and all $c_j\in C$.  
The goal is to compute a set of flips $\calS=(S_1, \dots, S_n)$
in the voters' preference tables 
that makes $p$ an $\calR$-winner of the election at minimum cost.  
We identify each flip with the pair of candidates being
flipped, so we have $S_i\subseteq C\times C$ for all $i=1, \dots, n$.
We denote an instance $M$ of $\calR$-\textsc{microbribery} by  
$(C,V,\Delta,p)$, where $\Delta$ is the sequence $(\delta^1, \ldots,
\delta^n)$ of microbribery price functions. Also, we 
let $\Delta(\calS)$ denote the total cost of flips in $\calS$, 
i.e., we set 
$$
\Delta(\calS)=\sum_{i=1}^n\sum_{(c_j, c_k)\in S_i}\delta^i(c_j, c_k).
$$ 

Faliszewski et al.~\cite{fal-hem-hem-rot:j:llull} have shown that   
Copeland$^\alpha$-\textsc{microbribery} is in $\p$ for $\alpha \in
\{0,1\}$. However, it is not clear if their algorithm can be extended
for all $\alpha\in[0, 1]\cap\Q$. Therefore, instead of using the
result of~\cite{fal-hem-hem-rot:j:llull} directly, 
we will now prove that Copeland$^\alpha$-\textsc{microbribery}
is in $\p$ for all $\alpha\in[0, 1]\cap\Q$ as long as the only permissible
flips are the ones that involve the manipulator's preferred
candidate $p$. We will then show that for all $\alpha\in[0, 1]\cap\Q$, 
finding an $m$-approximate solution for Copeland$^\alpha$-\textsc{shift-bribery} 
can be reduced to solving an instance of Copeland$^\alpha$-\textsc{microbribery}
that satisfies this constraint.

\begin{lemma}\label{lem:allalpha}
  For any $\alpha\in [0,1]\cap\Q$, 
  {\em Copeland$^\alpha$-}\textsc{microbribery} is in $\p$ 
  for all instances $M=(C, V, \Delta, p)$ with $|C|=m$, $|V|=n$ 
  such that for any $i=1, \dots, n$ and any two
  $c_j, c_k \in C$, it holds that if $p \not\in \{c_j,c_k\}$
  then $\delta^i(c_j, c_k)=+\infty$.
\end{lemma}
\begin{proof}
  Fix $\alpha\in[0,1]\cap \Q$, and let $M = (C,V,\Delta,p)$
  be our input instance of Copeland$^\alpha$-\textsc{microbribery}.   
  Assume that $C = \{p, c_1, \ldots, c_{m-1}\}$, and let $E=(C, V)$.  

  Our algorithm works as follows.
  For all $i=0, \dots, m$ and all $j=0, \dots, m$, 
  we compute $k(i, j)=i+\alpha j$, and check whether $\score_E(p) \leq k(i, j) \leq m-1$.
  If this condition is satisfied, we 
  find a minimum-cost microbribery $\calS(i, j)$
  that ensures that $p$ has $k(i, j)$ points and every other
  candidate has at most $k(i, j)$ points; if no such microbribery
  exists, we declare the cost of $\calS(i, j)$ to be $+\infty$.
  In the end, we output the cheapest microbribery found. 
  Observe that setting $i=m-1$, $j=0$ results
  in a successful microbribery of finite cost in which $p$
  wins every pairwise election. 
  Thus, we are guaranteed to output a microbribery of finite cost.
  This approach is clearly correct and it remains to show how to compute  
  the microbribery $\calS(i, j)$ for each pair $(i, j)$.

  Let us fix a pair $(i, j)$ such that $\score_E(p) \leq k(i, j) \leq m-1$. 
  Since we can
  only flip entries of preference tables that involve $p$, for each
  candidate $c_k \in C$ we can decrease $c_k$'s score by either $1$ or
  $1-\alpha$ (if $c_k$ wins the pairwise election with $p$) or by
  $\alpha$ (if $c_k$ ties the pairwise election with $p$).  In each
  case it is easy to compute the cheapest way of achieving this.
  Thus, for each $c_k \in C$, if $\score_E(c_k) > k(i, j)$, we perform the
  cheapest microbribery that involves preference table entries for $p$
  and $c_k$ and brings $c_k$'s score down to $k(i, j)$ or below; clearly, 
  if this is impossible, then for this pair $(i, j)$ a microbribery 
  with the required properties does not exist. 
  After this step, each candidate $c_k \in C\setminus\{p\}$ has at most $k(i, j)$
  points and it remains to find the cheapest microbribery that ensures
  that $p$ has exactly $k(i, j)$ points. This can easily be done by using 
  standard dynamic programming techniques. 
\end{proof}

We will now use Lemma~\ref{lem:allalpha} to 
obtain an $m$-approximation algorithm for {\sc shift-bribery} 
under Copeland$^\alpha$, where $\alpha \in [0,1] \cap \Q$.

\begin{theorem}\label{thm:copeland}
  There exists a poly-time algorithm that given an instance $I =
  (C,V,\Pi,p)$ of {\em Copeland$^\alpha$-}\textsc{shift-bribery} with
  $\alpha \in [0,1]\cap \Q$ and $|C|=m$ outputs a shift-action $\vecs$
  such that $p$ is a winner in $\shift(C,V,\vecs)$ and $\Pi(\vecs)
  \leq m\cdot\opt(I)$.
\end{theorem}
\begin{proof}
Fix an $\alpha \in [0,1]\cap\Q$ and an instance $I = (C,V,\Pi,p)$ of
Copeland$^\alpha$-\textsc{shift-bribery} with $C=\{p, c_1, \dots, c_{m-1}\}$, 
$V=(v^1, \dots, v^n)$. 
For each $i=1, \dots, n$, 
assume that the preference order of a voter $v^i$ is given by
$$
c_{j_{i,k(i)}} \succ c_{j_{i,k(i)-1}} \succ \cdots \succ c_{j_{i,1}} \succ p \succ \cdots.
$$

Our algorithm first converts $I$ into an instance $M = (C,\hat{V},\Delta,p)$ of
Copeland$^\alpha$-\textsc{microbribery}. The list $\hat{V}$
contains $n$ voters. The preference table of each voter $\hat{v}^i$
is constructed from the preference ordering of $v^i$: we set
the $(j, k)$-th entry of $\hat{v}^i$'s preference table to 1
if and only if $c_j \pref^i c_k$.

The price functions $\Delta = (\delta^1, \ldots, \delta^n)$ 
are defined as follows.  
For each $i=1, \dots, n$ and all $\ell=1, \dots, k(i)$, we set
$\delta^i(p,c_{j_{i,\ell}}) = \delta^i(c_{j_{i,\ell}},p) = \pi^i(\ell)$. 
For all other pairs of candidates, 
we set the value of $\delta^i$ to be $\infty$.

Note that the resulting instance of Copeland$^\alpha$-\textsc{microbribery} 
satisfies the conditions of Lemma~\ref{lem:allalpha}. 
Therefore, we can use  Lemma~\ref{lem:allalpha} 
to compute an optimal solution to $M$. 
It is easy to see that there is a solution to $M$ which does
not involve flipping any of the preference-table entries with cost
$\infty$.
Therefore, an optimal solution to $M$ is a  
sequence $\calS = (S_1, \ldots, S_n)$, where 
for each $i=1, \dots, n$ the set
$S_i$ consists of flips that involve $p$
and candidates in $c_{j_{i,1}}, \dots, c_{j_{i,k(i)}}$ only. 

We derive a shift-action $\vecs = (s_1, \ldots, s_n)$ from
$\calS$ as follows.  For each $i=1, \dots, n$, we set $s_i = 0$
if $S_i \cap \{c_{j_{i,1}}, \ldots, c_{j_{i,k(i)}}\} = \emptyset$, and
$s_i=\max\{\ell\mid c_{j_{i,\ell}} \in S'_i\}$ otherwise.  
It is easy to see that since $\calS$ is a
solution for $M$, $\vecs$ is a shift-action that ensures that $p$ is a
winner in $I$. Further, we have
\begin{equation} \label{eq:con:1}
  \Pi(\vecs) \leq \Delta(\calS).
\end{equation}
Our algorithm outputs $\vecs$. It remains to show that
$\Delta(\calS)\le m\cdot\opt(I)$.

Let $\vect = (t_1, \ldots, t_n)$ be an optimal-cost shift-action 
that ensures $p$'s victory in $I$. From $\vect$, we can derive a
solution $\calT = (T_1, \ldots, T_n)$ to $M$ as follows.
For each $i=1, \dots, n$, if $t_i=0$, we set $T_i = \emptyset$, 
and otherwise we set
$$
T_i = \{(p, c_{j_{i,1}}), (p, c_{j_{i,2}}), \ldots, (p, c_{j_{i,t_i}})\}.
$$
Since applying the shift-action $\vect$ to $I$ makes
$p$ a winner, applying $\calT$ to $M$ makes $p$ a winner in $M$
as well. Naturally, since $\calS$ is an optimal solution for $M$, we
have 
\begin{equation}\label{eq:con:2}
  \Delta(\calS) \leq \Delta(\calT).
\end{equation}
For each $i=1, \dots, n$, let $\delta^i(T_i) = \sum_{(p, c) \in T_i}\delta^i(p,c)$. 
We have
\begin{equation}\label{eq:con:3}
    \delta^i(T_i) =  \sum_{(p, c) \in T_i}\delta^i(p,c) 
                  =  \sum_{\ell=1}^{t_i}\delta^i(p,c_{j_{i,\ell}})
     = \sum_{\ell=1}^{t_i}\pi^i(\ell) \leq t_i\pi^i(t_i) \leq m\pi^i(t_i).
\end{equation}   
Consequently, $\Delta(\calT) \leq m\Pi(\vect)$. Combining this
  inequality with~\eqref{eq:con:1} and~\eqref{eq:con:2}, we obtain
\[
  \Pi(\vecs) \leq \Delta(\calS) \leq \Delta(\calT) \leq m\Pi(\vect).
\]
Since $\vect$ is an optimal solution for $I$, the shift-action $\vecs$ is an
$m$-approximate solution for $I$. The algorithm clearly works in
polynomial time.
\end{proof}
The proof of Theorem~\ref{thm:copeland} can also be adapted for 
maximin (using the result of~\cite{elk-fal-sli:c:swap-bribery}, 
where the authors show that maximin-{\sc microbribery} is in $\p$), 
and, more generally, for any other rule
that is defined for irrational voters, has a polynomial-time 
microbribery algorithm (at least for the instances constructed in the
proof), and satisfies a certain form of monotonicity (we will not
formalize this notion of monotonicity; essentially, the voting rule
has to guarantee that our algorithm's translations between
microbribery solutions and shift-bribery solutions keep $p$ a winner).

To extend these approximability results
for Copeland and maximin to the case of weighted voters, 
we use essentially the same strategy as in the proof of Theorem~\ref{thm:scoring2}.
That is, we first show that both 
{\sc weighted} Copeland-{\sc microbribery} and
{\sc weighted}  Maximin-{\sc microbribery} admit an FPTAS
as long as the input instances satisfy the condition of Lemma~\ref{lem:allalpha}.
We then use this FPTAS as an oracle inside the algorithm
described in the proof of Theorem~\ref{thm:copeland} 
to obtain an $m(1+\eps)$-approximation scheme for our problems.
We then use the idea presented in the proof of Theorem~\ref{thm:scoring2}
to convert this approximation scheme into an $m$-approximation algorithm.
As in the proof of Proposition~\ref{prop:approx}, 
the FPTAS proceeds by executing 
$\lceil\log \sum_{i=1}^n\sum_{j=1}^{m-1}\delta^i(p, c_j)\rceil$
iterations that correspond to different guesses of the cost
of the optimal microbribery. In iteration $i$, we set $\rho=2^{i-1}$, 
and round all bribery prices that do not exceed $\rho$
up to the nearest multiple of $\eps\rho$, 
where $\eps$ is the given error parameter;
all prices that exceed $\rho$ are set to $+\infty$.
We then find an optimal microbribery $\calS(\rho)$ for the rounded prices, 
and discard it if its cost is too large relative to $\rho$.
For the rounded prices, the optimal microbribery for both Copeland
and maximin is easy to find, and the cost of the best non-discarded
microbribery provides a good approximation to the cost 
of the optimal microbribery.

Interestingly, our FPTAS for Copeland$^\alpha$ only works
for the instances of {\sc weighted} Copeland$^\alpha$-{\sc microbribery}
that satisfy the condition of Lemma~\ref{lem:allalpha}.
Indeed we can show that, in general, 
{\sc weighted} Copeland$^\alpha$-{\sc microbribery} is inapproximable
up to any factor for all rational $\alpha<1$. This result
holds even if there are only three candidates.

\begin{proposition}
For all rational $\alpha<1$ and any $K$, there does not exist
a poly-time $K$-approximation algorithm for 
{\sc weighted} {\em Copeland$^\alpha$}-{\sc microbribery} 
unless $\p=\np$.
\end{proposition}

\begin{proof} 
The reduction is from {\sc Partition}.  An instance of {\sc
  Partition} is given by a list of $n$ positive integers $A = (a_1,
\dots, a_n)$. It is a ``yes''-instance if there is a set of indices
$J\subseteq\{1, \dots, n\}$ such that $\sum_{i\in J}
a_i=\frac{1}{2}\sum_{i=1}^n a_i$, and a ``no''-instance otherwise.

Given an instance $A=(a_1, \dots, a_s)$ of {\sc Partition}, 
we set $\sum_{i=1}^na_i=2B$, and
construct an instance of {\sc weighted} Copeland$^\alpha$-{\sc microbribery}
as follows.
There are $s+2$ voters $(v^1, \dots, v^s, v^{s+1}, v^{s+2})$
with weights $(a_1, \dots, a_s, B, B)$,
and three candidates $a$, $b$ and $p$, where $p$ is the briber's preferred
candidate. The voter's preferences are set as follows.
Voters $v^{s+1}$ and $v^{s+2}$ prefer $p$ to any other candidate;
further, $v^{s+1}$ prefers $a$ to $b$, and $v^{s+2}$ prefers $b$ to $a$. 
All other voters prefer $a$ and $b$ to $p$, and $a$ to $b$.
Further, for voters $v^{s+1}$ and $v^{s+2}$ the cost of any flip is $+\infty$, 
while for $i=1, \dots, s$ we have $\delta^i(a, b)=0$, 
$\delta^i(a, p)=\delta^i(b, p)=+\infty$.

In this instance, $p$'s score is $2\alpha$, $a$'s score
is $\alpha+1$, and $b$'s score is $\alpha$. The only way
for $p$ to win is to ensure that all candidates have $2\alpha$
points, i.e., there is a tie between $a$ and $b$.
Clearly, this can only be done if we started
with a ``yes''-instance of {\sc Partition}.
\end{proof}

For the remainder of this section let us return to unweighted shift
bribery and maximin. We have already shown that maximin-\textsc{shift
  bribery} can be solved in polynomial-time with a linear
approximation ratio, and now we will show that, in fact, a logarithmic
one is possible and asymptotically optimal. The most important
building block of our new maximin algorithm is a result of Caragiannis
et al.~\cite{car-cov-fel-hom-kak-kar-pro-ros:c:dodgson} proved in the
context of computing the Dodgson score. Thus, let us now describe this
result and then show how it can be applied to
maximin-\textsc{shift-bribery}.

Given an election $E = (C,V)$, the {\em Dodgson score} of a candidate $c \in C$ 
is the minimum number of positions by which $c$ needs to be shifted
upwards in the preference orders of the voters in $V$ to become 
a Condorcet winner. Observe that the Dodgson score of a candidate
is exactly the cost of shift bribery that makes $c$ a Condorcet winner, 
assuming that each unit shift has a unit cost.

It is known that determining whether a candidate is a winner in
Dodgson elections is an $\np$-hard
problem~\cite{bar-tov-tri:j:who-won} (in fact, it is
$\thetatwo$-complete~\cite{hem-hem-rot:j:dodgson}). Caragiannis et
al.~\cite{car-cov-fel-hom-kak-kar-pro-ros:c:dodgson} gave a
polynomial-time $O(\log m)$-approximation algorithm for computing
Dodgson scores.  In fact, their algorithm is somewhat more general and
the next theorem is a direct consequence of the results
in~\cite{car-cov-fel-hom-kak-kar-pro-ros:c:dodgson}.

\begin{theorem}
  \label{thm:dodgson}
  There is an algorithm $\calF$ that given an instance $I = (C,V,\Pi,p)$
  of $\calR$-\textsc{shift-bribery} with $|C|=m$
  and a sequence $k_1, \ldots, k_{m-1}$ of nonnegative integers, computes 
  a shift-action
  $\vect$ that has the following properties:  
  \begin{enumerate}
  \item For each $i=1, \dots, m-1$, $N_{\shift(C,V,\vect)}(p,c_i)
    \geq \min( N_{(C,V)}(p,c_i)+k_i, |V|)$.
  \item $\Pi(\vect) = O(\log m)\Pi(\vecs)$, where $\vecs$ 
    is a minimal-cost shift-action that satisfies the above condition.
  \end{enumerate}
\end{theorem}

We can now use Theorem~\ref{thm:dodgson} to give a polynomial-time
$O(\log m)$-approximation algorithm for maximin-\textsc{shift-bribery}.
The main idea of our approach is that, for a given instance $I$ of
maximin-\textsc{shift-bribery}, it is easy to provide a description of
an optimal shift-action in the format of Theorem~\ref{thm:dodgson}.

\begin{theorem}\label{thm:maximin}
  There exists a poly-time algorithm that
  given an instance $I = (C,V,\Pi,p)$ of
  {\em maximin-}\textsc{shift-bribery} with $|C|=m$
  outputs a shift-action $\vecs$ such that $p$ is a winner in
  $\shift(C,V,\vecs)$ and $\Pi(\vecs) =O(\log m)\cdot\opt(I)$.
\end{theorem}
\begin{proof}
  Let $I$ be our input instance as in the statement of the theorem and
  let $E = (C,V)$.  Our algorithm executes a series of iterations, one
  for each value $k$, $\score_E(p) \leq k \leq |V|$. The goal of the
  iteration for value $k$ is to find a shift-action $\vecr^k = (r^k_1,
  \ldots, r^k_n)$ that (a) ensures that the score of $p$ is at least
  $k$ and the score of each candidate $c \in C\setminus\{p\}$ is at most $k$,
  and (b) $\Pi(\vecr^k)$ is within $O(\log m)$ of the cost of an
  optimal-cost shift-action that achieves (a). Given the shift-actions
  produced in these iterations, the algorithm outputs one with the
  lowest cost. Since the optimal shift-action for $I$ is an optimal
  shift-action for one of the iterations, this algorithm outputs an
  $O(\log m)$-approximate solution for $I$.  It remains to show how to
  execute the iterations.
  
  Let us now fix value $k$, $\score_E(p) \leq k \leq |V|$, and
  describe the iteration for $k$.  Let us assume that $C = \{p, c_1,
  \ldots, c_{m-1}\}$.  Since a shift-action can move $p$ only, via
  applying some shift action $\vecr = (r_1, \ldots, r_n)$ to $E$ we
  can affect each of the values $N_E(p,c_1), \ldots, N_E(p,c_{m-1})$
  but neither of the values $N_E(c_i,c_j)$ for $c_i,c_j \in C$.

  Thus, shift-action $\vecr^k$ has to satisfy the following
  conditions.  First, for each $i=1, \dots, m-1$ we require
  $N_{\shift(C,V,\vecr^k)}(p,c_i) \geq k$. This guarantees that
  the maximin score of $p$ is at least $k$. Second, for each 
  $c_i \in C\setminus\{p\}$
  such that $\score_E(c_i) > k$, 
  we require $N_{\shift(C,V,\vecr^k)}(c_i,p) < k$,
  which is equivalent to demanding that
  $N_{\shift(C,V,\vecr^k)}(p,c_i) \geq |V|-k$.  This guarantees that
  the score of each candidate in $C\setminus\{p\}$ is at most $k$. We see that
  these two requirements can easily be phrased in terms of the input
  to the algorithm of Theorem~\ref{thm:dodgson}. Thus, we compute
  $\vecr^k$ using Theorem~\ref{thm:dodgson}. Clearly, the resulting 
  shift action $\vecr^k$ satisfies our requirements. 
  Since $\vecr^k$ can be
  computed in polynomial time, the proof is complete.
  \end{proof}

We remark that the approximation guaranteee
given by Theorem~\ref{thm:maximin}
is asymptotically optimal. This follows from the fact that the
reduction of \textsc{exact-cover-by-3-sets} to
maximin-\textsc{shift-bribery} given in~\cite{elk-fal-sli:c:swap-bribery}
can be modified to reduce from \textsc{set-cover}, in a way that
allows maximin-\textsc{shift-bribery} to inherit 
the inapproximability properties of \textsc{set-cover}
(see~\cite{raz:c:pcp} for
inapproximability results for \textsc{set-cover}).

\section{Inapproximability of Swap Bribery}
In contrast to shift bribery, swap
bribery is hard to approximate up to an arbitrary factor 
for all voting rules for which the possible winner problem is hard;
this includes almost all scoring 
rules, and, in particular, $k$-approval for $k\ge 2$
and Borda~\cite{bet-dor:c:possible-winner-dichotomy}, 
Copeland~\cite{xia-con:c:possible-necessary-winners}, and
maximin~\cite{xia-con:c:possible-necessary-winners}.
Indeed, in the reduction from the possible winner problem to swap bribery, 
the resulting instance of swap bribery has a bribery of cost $0$
if and only if the original instance of the possible winner
problem is a ``yes''-instance. Thus, {\em any} approximation
algorithm for swap bribery can be used to decide the possible winner problem.

\section{Conclusions}\label{sec:conclusions}
We have presented approximation algorithms 
for campaign management under a number of voting rules.
Most of our results hold even for weighted voters.
We believe that designing algorithms for the case
of weighted voters is important, since
in realistic campaign management scenarios a ``voter''
to be bribed is usually a group of voters that can be reached
by the same ad. By the same token, it would be interesting
to extend our results to settings where we can reach 
several {\em non-identical} voters with the same ad;
this would correspond to shift bribery with ``bulk discounts''.
Another, more applied direction would be to identify
commercial campaign management scenarios (where candidates
correspond to services or products) that can be handled 
using our model; the sponsored search example is the introduction 
is the first step in that direction.
Finally, a natural direction for further study is 
to design efficient algorithms for shift bribery with better 
approximation ratios, 
or to prove that our results are (asymptotically) optimal.

\medskip 
\noindent\textbf{Acknowledgments.} We would like to thank anonymous
WINE-2010 referees for their comments on the paper. We are also very
grateful to Ildik\'o Schlotter for very useful comments on the
manuscript and pointing out bugs in our Bucklin voting proofs (omitted
from this paper). 
This work was done in part during Piotr Faliszewski's visit to Nanyang
Technological University.
Edith Elkind is supported by Singapore NRF Research Fellowship
2009-08.  Piotr Faliszewski is supported in part by AGH University of
Science and Technology Grant no.~11.11.120.865, by Polish Ministry of
Science and Higher Education grant N-N206-378637, by Foundation for
Polish Science's program Homing/Powroty, and by NSF grant CCF-0426761.

\bibliography{grypiotr2006}

\appendix
\section{Analysis of a single-pass variant of $\calA$}\label{app:single-loop}
A counterintuitive property of the algorithm $\calA$ described
in Section~\ref{sec:scoring} is that it uses two {\bf for}-loops
instead of one. That is, it splits the money to be spent into two parts $\ell_1$
and $\ell_2$, greedily selects a shift-action $\vecs'$ that maximizes $p$'s score
given the budget $\ell_1$, and then greedily selects a shift-action $\vecs''$
for $\shift(C, V, \vecs')$ given the budget $\ell_2$.
It is natural to ask if a ``greedy'' single-pass algorithm given 
in Figure~\ref{fig:algA1} is still a 2-approximation algorithm for our problem.
  \begin{figure}[bth]
    \begin{center} 
      \begin{tabbing}
        123\=123\=123\=123\=123\=\kill
        \> \textbf{procedure} $\calG( C, V, \Pi, p )$\\
        \> \textbf{begin}\\
        \>\> Set $m = |C|$, $n = |V|$, $M = \sum_{i=1}^n\pi^i(m)$, $b = \infty$; \\
        \>\> \textbf{for} $\ell = 0$ \textbf{to} $M$ \textbf{do}\\
        \>\> \textbf{begin}\\
        \>\>\>  $I' = \bestbuy(I,\ell)$; \\
        \>\>\>  \textbf{if} $p$ is an $\calR$-winner in $I'$ 
                            \textbf{and} $\ell < b$ \textbf{then}\\
        \>\>\>\>  $b = \ell$;\\
        \>\> \textbf{end}\\
        \>\> \textbf{return} $b$; \\
        \> \textbf{end}
      \end{tabbing}
      \caption{\label{fig:algA1}Algorithm $\calG$.}
    \end{center}
  \end{figure} 

It is easy to see that each shift-action considered
by $\calG$ is also considered by $\calA$, so $\calG$
never produces a better solution than $\calA$.
However, the converse is not true.
Indeed, we will now show that the solution produced by
$\calG$ can be almost twice as expensive as the one 
produced by $\calA$.

\begin{theorem}\label{thm:single}
For any $\eps>0$, there is an instance 
$I = (C, V, \Pi, p)$
of {\em Borda}-\textsc{shift-bribery} such that
$\Pi(\calG(I))\ge (2-\eps)\Pi(\calA(I))$.
\end{theorem}
\begin{proof}
Given an $\eps>0$, set $k=\lceil\frac{1}{\eps}\rceil$, 
and let $T=2k$.
We will now construct an instance $I$ 
of Borda-\textsc{shift-bribery} with
$\opt(I) = 2kT$, such that $\calA(I)$ outputs an optimal solution,
while $\calG(I)$ finds a solution of cost $4kT-3k$.
As we have $4kT-3k > 4kT-4k = 2kT(2-\frac{2}{T}) > 2kT(2-\eps)$, 
this suffices to prove the theorem.

To construct the instance $I = (C,V,\Pi,p)$, we
set 
$C = \{p,c,a_1, \ldots, a_{4k}\}$ and 
$V = (v^1, \ldots, v^{4k+2})$. 
The voters have the following preferences. 
For $i=1, \dots, 2k$,
$v^{2i-1}$ has preference order 
$c \succ p \succ a_1 \succ a_2 \succ \cdots \succ a_{4k}$ 
and $v^{2i}$ has preference order 
$c \succ p \succ a_{4k} \succ a_{4k-1} \succ \cdots \succ a_{1}$.
Voters $v^{4k+1}$ and $v^{4k+2}$ have preference orders 
$c \succ a_{4k} \succ a_{4k-1} \succ \cdots \succ a_{1} \succ p$ and 
$p \succ a_{1} \succ a_{2} \succ \cdots \succ a_{4k} \succ c$, 
respectively.  Let $E = (C,V)$.  It is easy to verify that 
the candidates have the following Borda scores:
$\score_{E}(p) = 16k^2+4k+1$, 
$\score_{E}(c) = (16k^2+4k+1)+4k$, 
and
$\score_{E}(a_i) = 8k^2+2k+1$ for $i=1, \dots, 4k$. 
That is, $c$ beats $p$ by $4k$ points, 
but $p$ has more points than any other candidate.

We will now define the sequence $\Pi = (\pi^1, \ldots, \pi^{4k+2})$ of price
functions. Since the voters $v^1, \ldots, v^{4k}$ rank $p$ 
second, we can fully specify 
$\pi^i$ for $i=1, \dots, 4k$ by setting $\pi^i(1)=T$.
Similarly, since $v^{4k+2}$ ranks $p$ first, 
the function $\pi^i$ is fully described by setting
$\pi^{4k+2}(0) = 0$. 
The price function $\pi^{4k+1}$ is defined as follows. We
have $\pi^{4k+1}(0) = 0$ and for each $j$, $1 \leq j \leq 4k+1$,
it holds that:
    \[
      \pi^{4k+1}(j) - \pi^{4k+1}(j-1)  = \left\{
      \begin{array}{ll}
           0   & \text{if $j = 0$,}  \\
        T+1  & \text{if $j = 1$,}  \\
        T   & \text{if $2 \leq j \leq k$}\\
        T-2  & \text{if $j=k+1$}\\
        T-1  & \text{if $k+2 \leq j \leq 4k+1$}
      \end{array}
      \right.
    \]
    
It is easy to see that
minimal-cost solutions involve shifting 
our preferred candidate $p$ upwards by one
position in the preference orders of $2k$ voters among those in
$V' = (v^1, \ldots, v^{4k})$: each such shift gives $p$ one extra point
and takes away one point from $c$. Thus, $\opt(I) = 2kT$.
    
Algorithm $\calA$ successfully finds an optimal solution.  
Indeed, any shift-action of cost $kT$
can increase $p$'s score in $I$ by at most $k$ and
to achieve this, we need to shift $p$ upwards by one position in some $k$
votes chosen from $V'$. Let $I'$ be the instance obtained by this
action. In $I'$, the shift-action of cost $kT$
that maximizes $p$'s score 
is obtained by shifting $p$ upwards by one position
in the votes of the $k$ ``unshifted'' voters from $V'$. 
Thus, for $\ell_1 = \ell_2 = kT$, $\calA$ finds an optimal
solution.
 
On the other hand, algorithm $\calG$ outputs the shift-action
$(0,\ldots,0,4k,0)$ of cost $T+1 + (k-1)T + 3k(T-1)-1 = 4kT-3k$.  
Indeed, if $\ell \leq kT$, the subroutine $\bestbuy(I,\ell)$
shifts $p$ upwards by one position in the
preference orders of $\lfloor \frac{\ell}{T} \rfloor$ voters among 
$V'$. On the other hand, if $\ell > kT$, $\bestbuy(I,\ell)$ 
shifts $p$ upwards as much as possible given the budget
$\ell$ in the preference order of $v^{4k+1}$
(recall that $\bestbuy(I,\ell)$
selects the cheapest shift-action among those that maximize $p$'s score
without exceeding cost $\ell$).
Moreover, by Lemma~\ref{lem:any2}, the shift
$(0,\ldots,0,4k,0)$ makes $p$ a winner, 
as it increases its score by $4k$ points, 
whereas the optimal solution increases $p$'s score by $2k$ points.
\end{proof}

Note that for the family of instances given in the proof of 
Theorem~\ref{thm:single}, $\calA$ actually produces an optimal solution, 
so $\calG$ still provides a factor $2$ approximation in this case.
In fact, it is not clear whether $\calG$
is a 2-approximation algorithm for shift bribery with respect
to the class of all scoring rules.
Note, however, that $\calG$ may behave very differently
from $\calA$: in the proof of Theorem~\ref{thm:single}, 
the shift-actions produced by $\calG$ and $\calA$ are
completely disjoint. Thus, even if $\calG$ happens to be a
$2$-approximation algorithm for our problem, 
the proof that it achieves this approximation ratio is very 
likely to require additional insights.

\end{document}